\documentclass[12pt]{article}

\usepackage[utf8]{inputenc}
\usepackage[T1]{fontenc}
\usepackage{authblk} 
\usepackage{amsmath,amsfonts,amssymb,amsthm}
\usepackage{dsfont} 
\usepackage[dvipsnames]{xcolor} 
\colorlet{mygreen}{green!60!black}
\colorlet{mygrey}{white!60!black}
\usepackage{graphicx} 
\usepackage{float}  
\usepackage{booktabs} 
\usepackage{setspace} 
\usepackage{multicol} 
\usepackage{caption}
\usepackage{subcaption}
\usepackage{stmaryrd} 
\usepackage[
  colorlinks = true,
  linkcolor  = MidnightBlue,
  citecolor  = MidnightBlue,
  urlcolor   = MidnightBlue
]{hyperref}

\usepackage{wrapfig}
\usepackage{lscape}
\usepackage{rotating}
\usepackage{epstopdf}

\usepackage{tikz}
\usetikzlibrary{shapes,positioning} 
\usetikzlibrary{decorations.pathmorphing}
\usetikzlibrary{calc} 
\usetikzlibrary {arrows.meta} 

\newlength{\wirespace} 
\newlength{\wirespaceh} 
\newlength{\wirespaceht} 
\newlength{\wirespacehb} 
\newlength{\boxsize}
\newlength{\messpace} 
\usepackage{geometry} 
\usepackage[font=small,labelfont=bf]{caption} 
\usepackage{palatino} 

\theoremstyle{definition}

\theoremstyle{theorem}
\newtheorem{theorem}{Theorem}

\def \be {\begin{equation}} 
\def \ee {\end{equation}}
\def \bes {\begin{equation*}}
\def \ees {\end{equation*}}
\def \baa {\begin{align}}
\def \eaa {\end{align}}
\def \baas {\begin{align*}}
\def \eaas {\end{align*}}
\def \bea {\begin{eqnarray}}
\def \eea {\end{eqnarray}}
\def \beas {\begin{eqnarray*}}
\def \eeas {\end{eqnarray*}}

\newcommand{\ket}[1]{\rvert#1\rangle}
\newcommand{\bra}[1]{\langle #1\rvert}

\newcommand{\hil}{\mathcal{H}^}
\newcommand{\ketbra}[2]{| #1 \rangle \langle #2 |}
\newcommand{\comp}[1]{\overline{#1}}

\newcommand{\one}{\mathds{1}}

\newcommand{\bol}{\boldsymbol} 

\newcommand{\Uni}{\text{U}} 
\newcommand{\Uf}{\mathsf U} 
\renewcommand{\dim}{\text{Dim}} 
\newcommand{\calO}{\mathcal{O}}
\newcommand{\totalsys}{\mathfrak{U}} 
\newcommand{\Ut}{U_t} 

\newcommand{\gen}{\mathbf{gen}} 

\newcommand{\whole}{\mathfrak U} 

\title{Realism and the Inequivalence
 of the\\
  Two Quantum Pictures\thanks{Forthcoming in Alyssa Ney (ed.), \textit{Local Quantum Mechanics: Everett, Many Worlds, and Reality}. New York: Oxford University Press.}}
\author{Charles Alexandre B\'edard}
\affil{\small {École de technologie supérieure}\\
\footnotesize \emph{charles.alexandre.bedard@etsmtl.ca}}
\date{July 2026
}
\begin{document}
\maketitle
\begin{abstract}\noindent The standard claim that the Schrödinger and Heisenberg pictures of quantum mechanics are equivalent rests on the fact that they yield identical empirical predictions. This equivalence therefore assumes the instrumentalist worldview in which theories serve only as tools for prediction. Under scientific realism, by contrast, theories aim to describe reality. Whereas the Schrödinger picture posits a time-evolving wave function, the Heisenberg picture posits so-called descriptors, time-evolving generators of the algebra of observables. These two structures are non-isomorphic: descriptors surject onto but do not reduce to the Schrödinger state. Hence, under realism, the pictures are inequivalent. I argue that this inequivalence marks an opening toward a richer, separable ontology for quantum theory. On explanatory grounds, descriptors provide genuinely local accounts of superdense coding, teleportation, branching, and Bell inequality violations—phenomena that the Schrödinger framework does not explain fully locally.
\end{abstract}

\section{Introduction}
In the Heisenberg picture of unitary quantum mechanics, physical systems are described in a fully local and separable way, as shown by Deutsch and Hayden at the turn of the millennium~\cite{deutsch2000information}. 
This formulation resolves the apparent nonlocality in quantum teleportation and superdense coding, and accounts for a truly local branching, 
which underlies a local explanation of Bell violations.
None of this has a counterpart in the Schrödinger picture, 
yet it is routinely asserted that the two pictures are equivalent. 
\emph{They are not}.

In the Schrödinger picture, the universe is described by a time-evolving wavefunction. 
In the Heisenberg picture, it is instead described by time-evolving local generators of the algebra of observables.
Previous work has noted that distinct Heisenberg-descriptor assignments correspond to the same Schrödinger state~\cite{timpson2005nonlocality, wallacetimpson07, raymond2021local, bedard2021cost}.
This many-to-one correspondence cannot be dismissed as a mere representational redundancy.
Rather, it means that 
the stories told in each picture are not in one-to-one correspondence.
The Heisenberg description is richer, and it surjects onto the more restricted and observation-driven wave function of the Schrödinger picture.
Thus under \emph{scientific realism}, which posits the existence of a real world and considers theories as attempts to describe it, the two pictures are inequivalent.
%

Dirac~\cite{dirac1964foundations, dirac1965quantum} pointed out the inequivalence between the two pictures, declaring: “\emph{The Heisenberg picture is a good picture, the Schrödinger picture is a bad picture, and the two pictures are not equivalent.}” 
While I share this conclusion, 
Dirac was concerned with quantum electrodynamics, where for certain Hamiltonians, the Schrödinger picture admits no solution, not even approximate ones, and not even for the vacuum state. 
My argument, by contrast, does not involve quantum electrodynamics. 
A modest network of qubits suffices to expose the conceptual gap.

Moreover, retrospective analyses of the infancy of quantum mechanics rejected the early claims of equivalence between matrix~\cite{heisenberg1925quantentheoretische} and wave mechanics~\cite{schrodinger1926quantisierung}, since the `proofs' by Schrödinger~\cite{schrodinger1926verhaltnis} and Eckart~\cite{eckart1926operator} were later recognized as inadequate~\cite{hanson1961wave, muller1997equivalence1, muller1997equivalence2}. 
Yet the critics of these early claims of equivalence did not—as I do here—contest its modern form. 
For instance, according to Hanson, the equivalence was only established with Born's statistical interpretation, `which at last makes it a matter of indifference which algorithm one chooses to express his predictions'~\cite{hanson1961wave}.

Hanson’s reduction of a theory to an algorithm for making predictions reflects the philosophical stance known as \emph{instrumentalism}, a view held by many quantum physicists.
If the sole aim of science is prediction---if its only goal is to compute distributions of observed outcomes irrespective of how those outcomes come about---then the two pictures are indeed equivalent. Namely, they are \emph{instrumentally} equivalent. 
Under scientific realism, however, the verdict is different.

\section{The Instrumentalist Equivalence}\label{sec:instr}

This section revisits the standard presentation of the relationship between the two pictures---a staple of physicists’ training.
It operates on two levels: a shared mathematical framework from which both pictures are built, and the claim that they yield identical predictions.

The connection between Heisenberg's matrix mechanics~\cite{heisenberg1925quantentheoretische} and Schrödinger's wave theory~\cite{schrodinger1926quantisierung} required substantial mathematical groundwork, which culminated in the work of von Neumann~\cite{von1955mathematical}.
Below are the axioms necessary to relate the two pictures, where I leave aside the technicalities of infinite dimensionality.
\begin{itemize}
\item[\textbf{A1.}]States, denoted by $\ket \psi$, are unit vectors in a Hilbert space $\mathcal H$;
\item[\textbf{A2.}] Observables, denoted by $\mathcal O$, are self-adjoint operators on $\mathcal{H}$; 
\item[\textbf{A3.}] The dynamics, here denoted~$\Ut$ between time~$0$ and~$t$, is a unitary operator;
\item[\textbf{A4.}] 
Measurement predictions are given by the Born rule: for state $\ket{\psi}$ and observable $\mathcal{O}$, the expectation value of observed outcomes is $\bra \psi  \mathcal{O} \ket{ \psi}$.
\end{itemize}
As expressed above, the axioms are picture-agnostic: since dynamics are expressed independently of states and observables, there is no commitment to the evolution of either. 
From this mathematical machinery, both pictures can be constructed.
An initial state $\ket{\psi_0}$ and an initial algebra of observables $\{\mathcal O_0\}$ are fixed.
In the Schrödinger picture, the system is described by a time-evolving state, \mbox{$\ket{\psi_t} = \Ut \ket{\psi_0}$}, 
while the observables remain fixed.
In the Heisenberg picture, the system is described by time-evolving observables, \mbox{$\mathcal O_t =\Ut^\dagger \,\calO_0\, \Ut $}, alongside the fixed~$\ket{\psi_0}$.

The equivalence is usually taken to rest on a simple identity: both pictures give the same Born-rule expectation values (axiom~\textbf{A4}), namely

\vspace{7pt}
\begin{tikzpicture}
		\node at (0, 0){\textsc{~}};
		\node at (3.4, 0.4){\textsc{Schr\"odinger}};
		\node at (3.4, 0){\textsc{picture}};
		\node at (9.9, 0.4){\textsc{Heisenberg}};
		\node at (9.9, 0){\textsc{picture}};
\end{tikzpicture}
\vspace{-8pt}
\be\label{eq:instreq}
\bra{\psi_t} \,\calO_0\, \ket{\psi_t} \,=\, \bra{\psi_0}\, \Ut^\dagger \,\calO_0\, \Ut \,\ket{\psi_0} \,=\, \bra{\psi_0} \, \calO_t\,  \ket{\psi_0}\,.
\ee

In the received view, what declares the pictures equivalent is that they give rise to the same \emph{observable predictions}---not to isomorphic time-evolving descriptions of physical systems.  

Such a low bar for equating theories is the reflection of \emph{instrumentalism}, which has long prevailed in quantum theory. 
According to that philosophy, a theory is an apparatus, an instrument, whose sole purpose is to enable us to compute predictions of measurements---it is Hanson's algorithm.
Questions about how the world is and how it gives rise to what we measure are at best ignored.
At worst, they are threatened away---`\emph{shut up and calculate}'---or they are tabooed by enforcing doctrines such as the meaninglessness of what happens between preparation and measurement.

More pervasively, instrumentalism entrenches the idea that the mysteries of quantum mechanics must remain mysteries. 
It does so implicitly by promoting the idea that only on observations---on the \emph{seen}---do we have a firm handle, while simultaneously stigmatizing attempts at explaining the seen in terms of the \emph{unseen}, that is, at explaining physical reality.
This is
a recipe for the stagnation of foundational research.

\section{Realism}\label{sec:invsre}

In contrast, \emph{realism}~\cite{popper1996realism} holds that there is a real objective world out there, independent of people and their ideas about it.
Scientific inquiry proceeds by positing stories about the world---theories---and testing them.
The exercise is inherently fallible, yet it still commits to the idea that concepts and structures in our best theories do correspond to aspects of reality, whether these aspects are close to observations or not. 
No one has ever \emph{directly} observed a nuclear reaction, but we still accept their existence since our best theories imply that they exist.

What instruments measure and what observers perceive 
are themselves physical processes---no different in kind from the phenomena being measured.
Thus the realist worldview affirms the universality of physical theories: instruments and observers are neither outside nor at the center of a theory---why would they be?
They are, after all, other physical systems.
By defending the universality of unitary quantum theory and therefore treating measurements like other interactions, Everett~\cite{everett1973theory} restored the compatibility of quantum theory with scientific realism.
Everett’s key idea was not to \emph{posit} many worlds---which he instead \emph{derived}---it was to consider unitary processes to be universal. 
With this, he rejected the instrumentalist patchwork in favour of realism.

Scientific realism immediately implies that two theories which make the same predictions are not necessarily equivalent.
Rather, they are equivalent if an isomorphism relates their structures.
For instance, Lagrangian and Hamiltonian mechanics are related by such an isomorphism:
the Legendre transform identifies the tangent and cotangent bundles of the configuration manifold.
Thus, not only do the Lagrangian and Hamiltonian formalisms yield the same observations, but the descriptions as time-evolving points indexed by either $(q, \dot q)$ or by $(q,p)$ are bijectively related.
The central claim of this chapter is that Schrödinger- and Heisenberg-picture descriptions are \emph{not} related by such an isomorphism.

\section{Heisenberg-Picture Descriptors}\label{sec:heis}

If the time-evolving wave function is the Schrödinger-picture description of a physical system, what is the Heisenberg-picture description?

In this section, I explain the framework of Deutsch--Hayden descriptors in a way that extends beyond qubits, drawing from other expositions~\cite{bedard2020information, raymond2021local, bedard2021cost, tibau2023locality}.
For a complete and pedagogical guide to descriptors in the quantum computation setting, which is arguably the most accessible exposition, see~\cite{bedard2021abc}. For more in this volume, see the chapter by Kuypers~\cite{kuypers2024restoring}.
Readers who prefer to first explore the motivation for the formalism may wish to skip ahead to \S\ref{sec:expl}.

In the Heisenberg picture, the state vector remains fixed while observables evolve in time. 
Thus, the object describing physical systems must be tied to observables, and not, despite its name, to the Heisenberg `state'.
But each system has an uncountably infinite set of observables, so how can one meaningfully describe a system in this picture? 
One might propose to track only the time evolution of specific observables whose expectation values are of interest. However, this approach is narrow in scope and lacks the generality of the Schrödinger picture, where the time-evolving state encodes the expectation values of all observables at once, or in other words, the distributions of any possible measurement.

\subsection{Generators}

The key is that all observables can be obtained from a \emph{generating set}, namely, a set of operators whose adjoints, products, and linear combinations span the entire operator algebra.
The generating set can and should be chosen such that each generator acts non-trivially on one single system.
The generators acting on a given system (and on that system only) are then collected into a tuple of operators, the \emph{descriptor} of the system.

Let~$\mathfrak U$ denote the whole system under consideration, which I shall refer to as the \emph{universe}.
Let us first consider that~$\mathfrak U$ contains a system~$\mathfrak S_1$ which is a qubit, so~$\mathfrak S_1=\mathfrak Q$.
Accordingly, the total Hilbert space is~$\hil{\mathfrak U} \simeq \hil{\mathfrak Q} \otimes \hil{\comp{\mathfrak Q}}$, where $\hil{\comp{\mathfrak Q}}$ is the Hilbert space pertaining to all degrees of freedom of systems other than~$\mathfrak Q$.
Let the reference Heisenberg state be set to $\ket{\bol 0} \in \hil{\totalsys}$, where \mbox{$\bra{\bol 0} \sigma_z \otimes \one^{\comp{\mathfrak Q}} \ket{\bol 0} = 1$}.
With this choice of initialization, the $z$ observable of the qubit is said to be \emph{sharp} with eigenvalue~$+1$. In the Schrödinger picture of quantum computing, this corresponds to the qubit being initialized in~$\ket 0 = \,\ket{\hspace{-2pt}\uparrow_z}$.
The descriptor of~$\mathfrak Q$ at time $0$ is given by
$$
\bol q_1(0) = \gen^{\mathfrak Q} \otimes \one^{\comp{\mathfrak Q}}\,, 
$$
where $\gen^{\mathfrak Q}$ is any tuple of operators that can generate an operator basis acting on the qubit space, i.e. a basis of~$\mathcal L(\hil{\mathfrak Q}) \simeq \mathcal L({\mathbb C}^2)$.
If~$\{\ket 0, \ket 1\}$ is a basis of $\hil{\mathfrak Q}$---for definiteness, let us fix it to the eigenstates of $\sigma_z$---then $\gen^{\mathfrak Q}$ can be, for instance, 
the canonical operator basis itself~$\{\ket{j}\bra{i}\}_{i,j=0,1}$. The tuple $\gen^{\mathfrak Q}$ can also be the pair of Pauli operators
$
(\sigma_x, \sigma_z)
$, because they multiplicatively generate~$\sigma_y = i \sigma_x\sigma_z$ and $\one= \sigma_x \sigma_x$; and $\{\one, \sigma_x, \sigma_y, \sigma_z\}$ is basis of \mbox{$\mathcal L({\mathbb C}^2)$}.
This choice is convenient, as Pauli generators make the action of quantum gates on qubits particularly transparent.
Yet if minimality is the goal, in fact~$\gen^{\mathfrak Q}$ can even consist of a single operator,~$\ketbra{1}{0}$. Indeed, by taking the adjoint, we find~$(\ketbra{1}{0})^\dagger = \ketbra{0}{1}$, and then by multiplication we obtain~$\ketbra{0}{0} = (\ketbra{0}{1})\,(\ketbra{1}{0})$ and~$\ketbra{1}{1}= (\ketbra{1}{0})\,(\ketbra{0}{1})$.

Suppose that the universe $\mathfrak U$ contains a second system~$\mathfrak S_2$ of $d$-dimensional Hilbert space~$\hil{\mathfrak S_2}$, with $d < \infty$. 
As before, the Hilbert space of the universe can be factorized into any subsystem and its complement, e.g. $\hil{\mathfrak U} \simeq \hil{\mathfrak S_2} \otimes \hil{\comp{\mathfrak S_2}}$ (here~$\simeq$ denotes an isomorphism, and the order of tensor factors carries no significance).
The descriptor of~$\mathfrak S_2$ at time~$0$ is given by
$
\bol q_2(0) = \gen^{\mathfrak S_2} \otimes \one^{\comp{\mathfrak S_2}}\,, 
$
where $\gen^{\mathfrak S_2}$ is any tuple of operators that can generate a basis of~$\mathcal L(\hil{\mathfrak S_2}) \simeq \mathcal L({\mathbb C}^d)$.
If~$\{\ket k\}_{k = 0}^{d-1}$ is a basis of $\hil{\mathfrak S_2}$, then $\gen^{\mathfrak S_2}$ can be the canonical operator basis $\{\ket{j}\bra{i}\}_{i,j=0,1, \dots, d-1}$.
Another choice for~$\gen^{\mathfrak S_2}$ is~$\{\ket{j}\bra{0}\}_{j=1, \dots, d-1}$, or the single operator $a = \sum_{j=1}^{d-1} \sqrt j  \, \ket{j-1}\bra{j}$. 
In each of these cases, taking the adjoint of the generators and multiplying the obtained operators together yields an operator basis\footnote{Observe that $aa^\dagger, (aa^\dagger)^2, \dots, (aa^\dagger)^d$ form a basis of the diagonal operators. Similarly, $aa^\dagger a^\dagger, (aa^\dagger)^2a^\dagger, \dots, (aa^\dagger)^{d-1}a^\dagger$ form a basis of the first subdiagonal. And so on.}.

If the universe $\mathfrak U$ contains a system~$\mathfrak S_3$ of infinite-dimensional Hilbert space~$\hil{\mathfrak S_3}$, we have again~$\hil{\mathfrak U} \simeq \hil{\mathfrak S_3} \otimes \hil{\comp{\mathfrak S_3}}$.
Let~$\{\ket k\}_{k = 0}^{\infty}$ be a countable basis of~$\hil{\mathfrak S_3}$.
The tuple $\gen^{\mathfrak S_3}$ can be represented by the family $\{\ket{j}\bra{i}\}_{i,j=0,1, \dots}$, or $\{\ket{j}\bra{0}\}_{j=1, 2 \dots}$, or it can be the pair\footnote{Because $a$ is not a bounded operator, the construction given in the previous footnote would lose ground. Instead, we note that the canonical basis can be obtained as $\{(a^{\dagger})^j \,(\ketbra 00)\, a^i / \sqrt{i!j!}\}_{ij}$.} of operators 
$( a = \sum_{j=1}^{\infty} \sqrt j  \, \ket{j-1}\bra{j} \,,\, \ketbra 00 )\,.$ 
If~$\hil{\mathfrak S_4}$ is a rigged Hilbert space admitting a Dirac-orthonormal set of eigenvectors~$\{\ket x\}_{x \in \mathbb R}$, then~$\gen^{\mathfrak S_4}$ can be~$\{\ketbra yx\}_{x,y \in \mathbb R}$\footnote{These continuously labelled operators can be formalized with Schwartz distribution theory as sesquilinear forms on test functions, 
$\ket f \,, \ket g \mapsto \bra g \ketbra yx \ket f = g^*(y) f(x)$.}.

\subsection{Separability}\label{sec:sep}

The descriptor of a collection of systems is the collection of descriptors. 
Indeed, let~$\bol q_i$ and~$\bol q_j$ be the descriptors of systems~$\mathfrak S_i$ and~$\mathfrak S_j$ respectively. 
A descriptor for the composite system~$\mathfrak S_i\mathfrak S_j$ must have the ability to generate all operators acting non-trivially on~$\hil{\mathfrak S_i} \otimes \hil{\mathfrak S_j}$ (and trivially elsewhere).
The tuple~$(\bol q_i, \bol q_j)$ works perfectly fine:~$\bol q_i$ can be used to construct a basis of operators acting on~$\hil{\mathfrak S_i}$ (and trivially elsewhere); and likewise, $\bol q_j$ spans a basis of operators acting on~$\hil{\mathfrak S_j}$.
By taking products of operators from~$\bol q_i$ with operators from~$\bol q_j$, one obtains a basis for~$\mathcal L(\hil{\mathfrak S_i} \otimes \hil{\mathfrak S_j})$.
The collection $(\bol q_i,\bol q_j)$ is, therefore, a valid descriptor for the composite system, just as required.

\subsection{Evolution of Descriptors and Observables}
Let $\mathfrak S_i$ be a system with initial descriptor $\bol q_i(0)$.
The descriptor evolves in time like operators do in the Heisenberg picture.
If~$U$ denotes the evolution on the total system~$\totalsys$ between time~$0$ and time~$t$, then
\be\label{eq:evotot}
\bol q_i(t) = U^\dagger \bol q_i(0) U \,,
\ee
where the conjugation by $U$ affects all the operators of $\bol q_i(0)$.

The time-evolved descriptor $\bol q_{i}(t)$ can be used to calculate any time-evolved observable $\mathcal O(t)$ pertaining to~$\mathfrak S^i$. 
This is first recognized at time~$0$, where~$\bol q_{i}(0)$ can generate an operator basis, and therefore, by also taking linear combinations, it can generate the initial observable~$\mathcal O(0)$. 
The time-evolved observable~$\mathcal O(t)$ is obtained from the same generative process that constructed~$\mathcal O(0)$ from~$\bol q_{\mathfrak S_i}(0)$, but instead expressed in terms of~$\bol q_{\mathfrak S_i}(t)$. 
In other words, if~$f_{\mathcal O}$ is a function encoding the generation of~$\mathcal O(0)$ from~$\bol q_{i}(0)$, $\mathcal O(0) = f_{\mathcal O}(\bol q_{i}(0))$, then~
\be \label{eq:evobs}
\mathcal O(t) = f_{\mathcal O}(\bol q_{i}(t)) \,.
\ee
Eq.~\eqref{eq:evobs} can be shown as follows.
For any $g, g' \in \gen^{\mathfrak S_i}$, $g \otimes \one^{\comp{\mathfrak S_i}}$ and $g' \otimes \one^{\comp{\mathfrak S_i}}$ are \emph{components} of $\bol q_i$, and they evolve in time according to Eq.~\eqref{eq:evotot}.
Taking the adjoint, multiplying and taking linear combinations of the time-evolved components always keep the~$U^\dagger$ and~$U$ outside of the expression,
\beas
\left(U^\dagger (g \otimes \one^{\comp{\mathfrak S_i}})U\right)^{\dagger} &=& U^\dagger (g^\dagger \otimes \one^{\comp{\mathfrak S_i}})U\\
\left(U^\dagger (g \otimes \one^{\comp{\mathfrak S_i}})U\right) 
\left(U^\dagger (g' \otimes \one^{\comp{\mathfrak S_i}})U\right) 
&=& U^\dagger (gg' \otimes \one^{\comp{\mathfrak S_i}})U\\
\lambda U^\dagger (g \otimes \one^{\comp{\mathfrak S_i}})U
+ \sigma
U^\dagger (g' \otimes \one^{\comp{\mathfrak S_i}})U
&=& U^\dagger (\lambda g+ \sigma g' \otimes \one^{\comp{\mathfrak S_i}})U \,.
\eeas
In more abstract terms, any generative manipulation in~$f_{\mathcal O}$---whether taking adjoints, products, or linear combinations---commutes with the global conjugation by~$U$.
Therefore, $f_{\mathcal O}(U^\dagger  \bol q_{i}(0) U) = U^\dagger f_{\mathcal O}(  \bol q_{i}(0) ) U $.

\subsection{Recovering the Density Operator}
The separability of descriptors (\S\ref{sec:sep}) entails that any time-evolved observable that pertains to a collection of systems can be obtained from the time-evolved descriptors of those systems.
To connect with the more familiar language of the Schrödinger picture, the descriptors corresponding to a collection of systems permit the reconstruction of the density matrix pertaining to this collection of systems.
In particular, the global density matrix can be obtained from the collection of individual time-evolved descriptors.

This can be shown as follows. Let~$\{\ket{i}\}_{i=1}^{\dim \mathcal H^{\whole}}$ be a basis of $\mathcal H^{\whole}$. 
Via their generating properties, the collection of all descriptors at time 0, denoted $\bol q_\whole(0)$, can generate the operators~$\ketbra ji$ acting on~$\mathcal H^{\whole}$. Thus, for some function $f_{ij}$,
$$
f_{ij}(\bol q_{\whole}(0)) = \ketbra ji \qquad \text{and}\qquad
f_{ij}(\bol q_{\whole}(t)) = U^\dagger \ketbra ji U \,. 
$$
The expectation value~$\langle f_{ij}(\bol q_{\whole}(t)) \rangle$ gives the matrix elements of the global density operator (in the specified basis):
\beas
\bra{\bol 0} f_{ij}(\bol q_{\whole}(t)) \ket{\bol 0} & = & \bra{\bol 0} U^\dagger \ketbra ji U \ket{\bol 0}\\
&=& \bra i U \ketbra{\bol 0}{\bol 0} U^\dagger \ket j\\
&=& \bra i \ketbra{\Psi(t)}{\Psi(t)} \ket j \\
&=& \bra i \rho_{\mathfrak U} \ket j\,.
\eeas

\subsection{Evolution of Evolutions}
In the Heisenberg picture, evolution operators can themselves be expressed as functions of evolving descriptors. Since they are, like observables, linear operators, they too can be reconstructed from descriptors. 
Suppose that between the discrete times $t-1$ and $t$, a localized operation on a possibly joint system~$\mathfrak S$ is performed---a quantum gate on~$\mathfrak S$. Let~$G_t$ denote the matrix representation of the operation on the \emph{whole} system~$\totalsys$, keeping in mind that~$G_t$ acts trivially on $\comp{\mathfrak S}$. Moreover, let $V$ be the evolution of~$\totalsys$ from time~$0$ to~$t-1$, so that,~$U= G_t V$. The evolution of descriptors can be expressed in a step-by-step fashion, relating their expression at time~$t$ with the one at time $t-1$. The descriptor of some system~$\mathfrak S_i$ at time $t$ is
\be \label{eq:step}
\bol q_i(t) = \Uf^\dagger_{G_t}(\bol q_{\mathfrak S}(t-1)) \bol q_i(t-1) \Uf_{G_t}(\bol q_{\mathfrak S}(t-1)) \,,
\ee
where~$\Uf_{G_t}(\cdot)$ is a fixed operator-valued function analogous to the~$f_{\mathcal O}$ encountered above. 
The function~$\Uf_{G_t}$ is defined by the requirement that \mbox{$\Uf_{G_t}(\bol q_{\mathfrak S}(0))=G_t$}, which is guaranteed to exist by the generative ability of $\bol q_{\mathfrak S}(0)$ to construct any linear operator acting non-trivially on ${\mathfrak S}$ (and so in particular, any unitary operator).

The expressions \eqref{eq:evotot} and \eqref{eq:step} for the evolution of $\bol q_i(t)$ are equivalent: 
\beas
 V^\dagger G_t^\dagger  \bol q_i(0) G_t V 
&=& V^\dagger \Uf^\dagger_{G_t}(\bol q(0)) V V^\dagger \bol q_i(0)V V^\dagger \Uf_{G_t}( \bol q(0) ) V \\
&=& \Uf^\dagger_{G_t}\left (V^\dagger \bol q(0) V \right ) V^\dagger \bol q_i(0) V \, \Uf_{G_t}\left(V^\dagger \bol q(0) V \right)\\
&=& \Uf^\dagger_{G_t}\left(\bol q(t-1) \right) \bol q_i(t-1) \Uf_{G_t}(\bol q(t-1)) \,.
\eeas
The second equality follows for the same reason as Eq.~\eqref{eq:evobs} holds; namely because in each term of the function $\Uf^\dagger_{G_t}\left (V^\dagger \bol q(0) V \right )$, 
products will have their inner $V^\dagger$s and $V$s cancelled, leaving only the outer ones, which can be factorized outside of the polynomial to retrieve the first line.

\subsection{No Action at a Distance}
Descriptors avoid action at a distance. 
To see this, consider, as in Eq.~\eqref{eq:step}, the evolution of some descriptor~$\bol q_i$ under the action of a gate~$G_t$ that affects system~$\mathfrak S$. 
However, let us assume here that $\bol q_i$'s system,~$\mathfrak S_i$, is not part of~$\mathfrak S$. Hence the gate~$G_t$, which does not affect~$\mathfrak S_i$, should leave~$\bol q_i$ invariant.
Let us verify this explicitly.
At time $0$, the descriptors of the two disjoint subsystems take the form
$$
\bol q_{i}(0) = \gen^{\mathfrak S_i} \otimes \one^{\mathfrak S} \otimes \one^{\comp{\mathfrak S_i\mathfrak S}} 
\qquad \text{and} \qquad
\bol q_{\mathfrak S}(0) = \one^{\mathfrak S_i}\otimes \gen^{\mathfrak S} \otimes \one^{\comp{\mathfrak S_i \mathfrak S}}\,. 
$$
It follows immediately that all components of~$\bol q_i(0)$ commute with all components of~$\bol q_{\mathfrak S}(0)$. 
Because commuting operators also commute when they are both conjugated by the same unitary operator, this commutativity is preserved in time. 
Therefore, in Eq.~\eqref{eq:step}, $\bol q_i(t-1)$ commutes with~$\Uf_{G_t}(\bol q_{\mathfrak S}(t-1))$, and the equation reduces to $\bol q_i(t) = \bol q_i(t-1)$.

Only with a separable description can we have a crisp case for no action at a distance.
For instance, according to Wallace~\cite{wallace2012emergent} `Action at a distance occurs when, given two systems~$A$ and~$B$ which are separated in space, a disturbance to~$A$ causes an immediate change in the state of~$B$, without any intervening dynamical process connecting~$A$ and~$B$'.
But what, exactly, is meant here by `the state of $B$'?
It cannot mean the total wave function, for that is always altered by any disturbance to~$A$.
And if it means the reduced density matrix~$\rho_B$, then no action at a distance collapses into the weaker condition of no-signalling.
Moreover, since the state of~$AB$ is generally more than the mere collection of~$\rho_A$ and~$\rho_B$, reduced density matrices provide only an incomplete account of systems and thus cannot fully adjudicate questions of locality. 
An operation on~$A$ may alter aspects of the joint state that are not captured by~$\rho_B$, and with no commitment as to where those aspects reside, the invariance of~$\rho_B$ offers no guarantee of locality. 
Related concerns are raised by Waegell and McQueen in this volume~\cite{waegell2025nonlocal}.
By contrast, the separable and complete description~$(\bol q_A,\bol q_B)$ leaves no residual structure with an unspecified location; any influence of~$A$ on~$B$ would have to appear directly in~$\bol q_B$.

\section{The Realist Inequivalence}\label{sec:realineq}

In this section, 
I first establish the one-to-one correspondence between descriptors and equivalence classes over unitary operations.
This serves to establish (\S~\ref{sec:noniso}) that the universal descriptor of a system cannot be put in one-to-one correspondence with the universal wave function.
 I then explore more generous descriptions that may make the Schröinger picture bijectively related to Heisenberg-picture descriptions~(\S\ref{secmore}).

\subsection{Non-isomorphic State Spaces}\label{sec:noniso}

Upon formalizing and axiomatizing local realism, Raymond-Robichaud~\cite{raymond2017equivalence} showed that any non-signalling theory whose set of operations forms a group can be lifted to a local-realistic theory.
In quantum theory, no-signalling is a property at the level of reduced density matrices, whereby actions on remote systems must leave a given density matrix unchanged.
Raymond-Robichaud's construction then gives a deeper layer of description, \emph{quantum noumenal states}, which fulfils his axiomatization of local realism.

The quantum noumenal state of a system $ \mathfrak S_k$ is defined as an equivalence class. It is the set of dynamics that differ from $U$ only in a way that do not causally concern $\mathfrak S_k$,
\be\label{eqqns}[U]^{\mathfrak S_k}= \left \{U' \in \Uni (\hil{\totalsys}) ~:~   U' = (\one^{\mathfrak S_k} \otimes W) U \text{ for some } W \in \Uni (\hil{\comp {\mathfrak S_k}}) \right \} \,,
\ee
where~$\Uni (\hil{})$ denotes the unitary operators on~$\hil{}$.
In the following theorem, first proven for qubits in Ref.~\cite{bedard2021cost}, I show that quantum noumenal states correspond one-to-one with descriptors.

\begin{theorem}[] \label{thmequiv3}
Let~$\totalsys$ be the whole system considered, with Heisenberg reference vector \mbox{$\ket{\bol 0} \in \hil{\totalsys}$}.
Assume that the whole Hilbert space admits, for a suitable set of indices $I$, the following decomposition
$$
\hil{\totalsys} = \bigotimes_{i \in I} \hil{\mathfrak S_i} \,,
$$
where~$\hil{\mathfrak S_i}$ has dimension~$d_i \in \mathbb N_{>1} \cup {\infty}$.
For all possible pairs of evolution~$U$ and~$U'$ of~$\totalsys$,
\bes
[U]^{\mathfrak S_i} = [U']^{\mathfrak S_i}~\iff \bol q_i(t) = \bol q'_i(t) \,,
\ees
where $\bol q_i (t)= U^\dagger \bol q_i(0) U $ and $\bol q'_i(t) = U'^\dagger \bol q_i(0) U' $.
\end{theorem}

\begin{proof}

First, let $[U]^{\mathfrak S_i} = [U']^{\mathfrak S_i}$, namely, $U' =(\one^{\mathfrak S_i} \otimes W) U$.
\beas \label{eq:preuvepaul}
\bol q_i' (t)
&=& U'^\dagger \left( \gen^{\mathfrak S_i}\otimes \one^{\comp{{\mathfrak S_i}}} \right) U' \nonumber \\
&=& U^\dagger (\one^{\mathfrak S_i} \otimes W^\dagger) \left( \gen^{\mathfrak S_i}\otimes \one^{\comp{{\mathfrak S_i}}} \right) (\one^{\mathfrak S_i} \otimes W)U \nonumber \\
&=& U^\dagger \left( \gen^{\mathfrak S_i}\otimes \one^{\comp{{\mathfrak S_i}}} \right) U \nonumber \\
&=& \bol q_i (t) \,. 
\eeas

 To prove the other implication, `$\Longleftarrow$', assume $[U]^{\mathfrak S_i} \neq [U']^{\mathfrak S_i} $ and therefore,
 $U' \neq (\one^{\mathfrak S_i} \otimes W) U$, for some~$W$ acting on~$\comp{\mathfrak S_i}$. Hence, $U' = V U$, for some global operator $V$, whose functional representation~$\Uf_V(\bol q(0))$ depends explicitly on terms of $\bol q_i (0)$. But then, if $V$ is thought to occur between time $t$ and $t+1$,
\beas
 \bol q_i (t+1) &=& U^\dagger V^\dagger \bol q_i(0) V U \\
&=& U^\dagger V^\dagger U U^\dagger \bol q_i(0) U U^\dagger V U \\
&=& \Uf^\dagger_V(\bol q(t)) \bol q_i(t) \Uf_V(\bol q(t)) \,.
\eeas
But because of its dependence on $\bol q_i(t)$, $\Uf_V(\bol q(t))$ acts nontrivially on~$\bol q_i(t)$ which means $\bol q_i(t+1) \neq \bol q_i(t)$.
\end{proof}

This theorem shows that the descriptor of a system encompasses the part of the unitary dynamics that is in the backward light cone of the system. When that system is $\totalsys$ as a whole, 
$$\bol q_{\totalsys}(t) \simeq [U]^\totalsys = U ~\text{(up to a phase)}\,.$$

A more pedestrian approach can also be used to establish that the local descriptors of all systems provide the knowledge of the evolution operator~$U$, up to a phase.
Indeed, from the descriptors of each system, one can generate a canonical basis $\{\ketbra{j}{i}\}$ of linear operators acting on~$\mathcal H_\totalsys$, where $i$ and $j$ are appropriate labels for the total Hilbert space. 
When time-evolved, this basis is $\{U^\dagger \ketbra{j}{i} U\}_{ij}$. 
The matrix element~$\ell,k$ of~$U^\dagger \ketbra{j}{i} U$ is given by
\bes
\bra \ell U^\dagger \ketbra{j}{i} U \ket k = u^*_{j\ell} u_{ik} \,.
\ees
 By setting~$i=j=k=\ell=0$, one finds~$|u_{00}|^2$, which can be assumed to be non-zero by otherwise permuting the columns of~$U$. By setting~$j=\ell=0$, but leaving~$i$ and~$k$ free, one finds~$u^*_{00} u_{ik}$ for all~$i$ and~$k$. Therefore, up to a phase ($u^*_{00}/|u_{00}|$), $U$ can be computed from $U^\dagger \ketbra{j}{i} U$ for all~$i$ and~$j$, which can be computed from $\bol q_i(t)$ for all~$i$.

 With this equivalent representation of descriptors in hand, we can easily recognize that they are not isomorphic to state vectors.

The descriptor state space is given by
\be\label{eq:desspace}
\mathsf{H}\text{-}\mathsf{Descriptors}^\totalsys = \{ U^\dagger \bol q^{\totalsys}(0) U : U \in \Uni (\hil{\totalsys})\}\,.
\ee
As we have seen, this is isomorphic to $\Uni (\hil{\totalsys}) / \Uni(1)$, or equivalently, to the projective unitary group $\mathcal P (\Uni (\hil{\totalsys}))$.

On the other hand, the Schrödinger-picture state space is given by
the equivalence classes of unit vectors under the equivalence relation
$\ket\psi  \hspace{-2pt}\sim  \,\,\ket \phi$ if and only if $ \ket \psi = e^{i\theta} \ket \phi$.
To express this with familiar objects,
$$
\mathsf{S}\text{-}\mathsf{States}^{\totalsys} 
= \{U \ket{\bol 0} : U \in \Uni (\hil{\totalsys})\} \,/ \sim \,.
$$
This space corresponds to the \emph{projective Hilbert space}~$\mathcal{P}(\hil{\totalsys})$.
Therefore, there is no isomorphism between these spaces, as
$$
\mathsf{H}\text{-}\mathsf{Descriptors}^\totalsys  \simeq \mathcal P(\Uni (\hil{\totalsys})) \not \simeq \mathcal{P}(\hil{\totalsys}) \simeq \mathsf{S}\text{-}\mathsf{States}^{\totalsys} \,.
$$
In other words, the collection of descriptors $\{\bol q_i(t)\}_{i = 1,...,n}$ contains all the information about the whole unitary (up to a phase).
On the other hand, the global state encompasses only a part of the unitary: essentially one column in some basis, because~\mbox{$\ket{\psi(t)} = U \ket{\bol0}$}.
For advocates of Schrödinger-picture Everettian quantum theory, the world is described by the universal wave function.
That is a thinner description than what is provided by the collection of descriptors. 

\subsection{Larger Schrödinger-Picture Descriptions?}\label{secmore}

It may be suggested that to properly describe systems in the Schrödinger picture, one should specify more than the global state alone. The dynamics also matter. One natural move would be to append the state~$\ket{\Psi(t)}$ with the unitary operator~$U$ that generated it, thereby enriching the Schrödinger description with explicit dynamical information. But how should~$\ket{\Psi(t)}$ and~$U$ be located within the subsystem structure?
Strapping the full pair~$(\ket{\Psi(t)},U)$ onto each subsystem would manifestly violate no action at a distance, 
since any local gate would alter the appended~$U$ (and hence~$\ket{\Psi(t)}$) simultaneously for all systems.

Hence, if we are to parallel the Heisenberg description, we must seek a genuinely local expansion of the Schrödinger picture. Instead of a single global object, we would need a collection $\{s_i(t)\}_i$ of subsystem-specific descriptors, each $s_i(t)$ playing the role of a Schrödinger-side analogue to the Heisenberg $\bol q_i(t)$. Such a construction, if it exists, would provide an isomorphism at the level of descriptions.

Mathematically, the structure of each $s_i(t)$ would need to be recoverable from the global unitary evolution, together with the initial state and observables, just as Heisenberg descriptors are. Yet here lies the difficulty: when we examine concrete proposals for such $s_i(t)$, what emerges looks nothing like the familiar Schrödinger picture.

One candidate is the quantum noumenal state $[U]^{\mathfrak S_i}$ encountered previously (see Eq.~\eqref{eqqns}). It is the equivalence class of global unitaries from $0$ to $t$, modulo those operations that lie outside the causal past of system~$\mathfrak S_i$. Relative to a fixed reference state, it captures precisely the unitary history that could influence $\mathfrak S_i$ at time $t$.

Another proposal is Waegell’s `local fluids in spacetime' framework~\cite{waegell2023local}, in which each system is described by its internal memory. This memory includes both the reference state $\ket{\bol 0}$ and the complete causal record of all interactions in the system’s past light cone. Unlike descriptors or noumenal states, which compress the past, internal memories retain it in full detail. For example, if a gate and its inverse are applied in sequence, the memory records both, whereas the descriptor erases the redundancy.

Both noumenal states and internal memories thus offer locally defined structures in spacetime, parallel in spirit to Heisenberg descriptors. But by dispensing altogether with the evolving state vector, they depart so radically from the traditional Schrödinger picture that they cannot plausibly be regarded as its reformulation.

\section{Additional Explicans}\label{sec:expl}

As I demonstrated, the Heisenberg-picture description of a quantum system stands in a many-to-one correspondence with its Schrödinger-picture state. 
The richer structures encompassed by descriptors offer a separable account of quantum systems, which in turn gives a precise notion of no action at a distance, which it also satisfies.
Consequently, descriptors fulfill what Kuypers calls, in this volume~\cite{kuypers2024restoring}, \emph{the principle of locality}.
In this section, I survey the consequences of the locality of descriptors, without pursuing the formal details, for which I point to further reading.

\subsection{Local Superdense Coding}
\emph{Superdense coding}~\cite{bennett1992communication} is a quantum information protocol which permits the sending of two classical bits, $i$ and $j$, by transmitting only one qubit; assuming that a pair of qubits in a known Bell state is shared by the sender (Alice) and the receiver (Bob).
%
In the Schrödinger picture, the phenomenon is explained as follows: by affecting her qubit in one of four ways via the Pauli operations $\sigma_z^i \sigma_x^j$,
Alice alters the entangled state to any one of the four Bell states. 
Should Alice's qubit be intercepted while being transmitted, the bits $i$ and $j$ cannot be retrieved from the qubit alone, since, regardless of which Bell state it is, the corresponding reduced density matrix is completely mixed. 
Thus, one might posit that the information about $i$ and $j$ resides in some global properties of the entangled pair. 
When Bob receives Alice's qubit, he measures the pair in the Bell basis and retrieves the two bits, $i$ and $j$.

The explanation in terms of descriptors is very different.
When Alice performs the operations $\sigma_z^i \sigma_x^j$, the bit~$i$ gets encoded in the $x$-component of her descriptor and the bit~$j$ in the $z$-component.
As it should be, nothing changes on Bob's side; the bits are localized within Alice's system and within Alice's system only. 
Yet, no measurement performed on that system alone can reveal information about~$i$ and~$j$. 
This is because the information is \emph{locally inaccessible}---it is encoded in Alice's descriptor, yet it can only be retrieved upon interacting with Bob's system. 
In the descriptor of its qubit, Bob holds the key to render $i$ and $j$ accessible after he receives Alice's qubit.
See Ref.~\cite[\S6]{bedard2021abc} for more details.

\subsection{Local Teleportation}\label{sec:teleportation}

\sloppy
\emph{Quantum teleportation}~\cite{bennett1993teleporting} is a quantum information protocol in which Alice transmits the state of a qubit by communicating two bits of classical information, and using shared entanglement.
Even if one treats measurement unitarily, the Schrödinger picture lacks a local explanation of the phenomenon; the qubit appears to be~`teleported'.
This is because the state vectors themselves are inadequate for localizing information. 
The complex parameters encoding the qubit's state at Alice's location are not tied to Alice's Hilbert space:
The state vector can be equivalently expressed with the parameters residing on Bob's system, albeit masked by Pauli operators.
Bob corrects these operators after receiving the classical bits, completing the teleportation.
See Ref.~\cite[\S2]{bedard2023teleportation} for the calculations.

Expressing the situation in terms of descriptors reveals how the quantum information is transported fully locally \emph{by the classical bits}. 
But how can classical information transport quantum information? 
We may recall Everett here, for there is, strictly speaking, no such thing as purely classical information. In a unitary framework, what we call “classical” is only the quantum made to look classical—an appearance that must be explained from within quantum theory itself.
Accordingly, teleportation remains successful under decoherence in the communication channel, or when the channel consists of a cascade of intermediate systems.
These are desirable properties of communication processes we might want to call `classical' in a fundamentally quantum world.
See 
Ref.~\cite{bedard2023teleportation} for a detailed discussion of teleportation.

\subsection{Truly Local Branching}

The discontinuous, non-local, logically irreversible and fundamentally stochastic collapse was shown to be illusory by Everett.
He did so by demonstrating how, in all respects where the collapse was deemed empirically necessary, unitary evolutions were in fact sufficient. 
Such empirical facts include the apparent irreversibility of measurements, the apparent uniqueness of measurement outcomes, their unpredictability, and their stability under repeated measurements and across observers.

The prerequisite to defending locality in quantum theory—the overarching theme of this volume—is to dispense with the dynamical non-locality of collapse. In its place, unitary quantum theory has \emph{branching}, the process by which systems evolve into distinct and autonomous entities.
In the Schrödinger picture, branching occurs when the wave function evolves into a sum of distinct relative states. These states remain autonomous because surrounding systems become entangled with the measured system, thereby proliferating records of the outcomes and preventing further interference.
Yet advocates of Everett in the Schrödinger picture disagree on whether and how branching is local.  
It suffices for my purposes to criticize the account I am most sympathetic to, the so-called `local branching' in the Schrödinger picture as put forth by Wallace~\cite[Chapter~8]{wallace2012emergent} and further discussed in this volume by Blackshaw, Huggett and Ladyman~\cite{BHLEverettian}.
And I shall criticize it on the basis that it is, in fact, not local.

Let us consider two particles entangled in their spin degrees of freedom, so that up to normalization their joint state is \mbox{$\ket{\hspace{-3pt}\uparrow}_1\ket{\hspace{-3pt}\uparrow}_2 \,+\,
 \ket{\hspace{-3pt}\downarrow}_1\ket{\hspace{-3pt}\downarrow}_2$}. Let~$\ket{\text{\footnotesize{Ready}}}_{\hspace{-1pt}A}$ denote a ready state for Alice and her measurement apparatus, and likewise~$\ket{\text{\footnotesize{Ready}}}_{\hspace{-1pt}B}$ for Bob and his apparatus.
Suppose that the measurements performed by Alice and Bob, respectively on Particles~1 and~2, happen at spacelike separation.  
The following unitary evolution relates the global state on spacelike hypersurfaces before and after both measurements;
$$
\left(\vphantom{A^{A}_A} \ket{\hspace{-3pt}\uparrow}_1\ket{\hspace{-3pt}\uparrow}_2 \,+\,
 \ket{\hspace{-3pt}\downarrow}_1\ket{\hspace{-3pt}\downarrow}_2  \right) \,
 \ket{\text{\footnotesize{Ready}}}_{\hspace{-1pt}A} \,
 \ket{\text{\footnotesize{Ready}}}_{\hspace{-1pt}B} 
~ \to ~
\ket{\hspace{-3pt}\uparrow}_1\ket{\hspace{-3pt}\uparrow}_2 \ket{\text{`$\uparrow$'}}_A \ket{\text{`$\uparrow$'}}_B 
~+~
 \ket{\hspace{-3pt}\downarrow}_1\ket{\hspace{-3pt}\downarrow}_2  \ket{\text{`$\downarrow$'}}_A \ket{\text{`$\downarrow$'}}_B  \,.
$$ 
In the above equation,~$\ket{\text{`$\uparrow$'}}_A$ and~$\ket{\text{`$\downarrow$'}}_A$ denote states of Alice and her measurement apparatus recording respectively the up and down outcome; and analogously for~$\ket{\text{`$\uparrow$'}}_B$ and~$\ket{\text{`$\downarrow$'}}_B$.

Taking the wave function at face value, it describes two distinct branches. The existence of two branches corresponding to Alice's possible outcomes is unproblematic, and the same applies to Bob.
What is puzzling, however, is that the branches extend across spacelike-separated regions, and that they already identify outcomes before any comparison has occurred.
According to the first term of the wave function, the Alice who measured~`$\uparrow$' is in the same branch as the Bob who also measured~`$\uparrow$', even though these measurements occurred at spacelike separation, and no physical interaction or comparison has yet occurred.
%
What mechanism enforces this nonlocal identification of outcomes, if branching is assumed to occur locally?

With descriptors, this difficulty does not arise.
When Alice measures her particle, Alice's descriptor evolves into a sum of two relative descriptors, each of which indicates a definite outcome.
The particle’s descriptor is also affected by the measurement, but nothing else changes.
Bob and all other systems not involved in Alice’s measurement remain completely unaffected, in line with the principle of locality.
When Bob measures his particle, he likewise evolves locally into two instances. 

Crucially, at this stage, the Alice who measured~`$\uparrow$' is not yet identified with the Bob who measured~`$\uparrow$'; the sets of branches are generated independently and locally.\footnote{From Alice’s and Bob’s relative descriptors at this point, one can predict what would happen if they later generated a joint record of outcomes. But this predictive power merely reflects determinism: future configurations follow from past states and the dynamics.
}
Only when Alice and Bob later compare results—an interaction that must itself be treated quantum mechanically—do the local branches merge into common branches. This comparison is crucial for explaining Bell locally, to which I now turn. See Refs.~\cite{kuypers2021everettian, kuypers2024restoring} for further discussion of local branching with descriptors.

\subsection{Local Violations of Bell Inequalities}

Some advocates of Everettian quantum mechanics invoke the multiplicity of outcomes in measurements as a way out of Bell’s theorem~\cite{bell1964}. And indeed, an assumption made by Bell is that measurements have a unique outcome, so allowing multiple outcomes blocks the theorem at the outset. But simply noting this does not yet explain how the coexistence of outcomes operates so as to reproduce the violations of Bell inequalities observed in experiments, let alone how it does so locally.

For each pair of CHSH inputs, Alice and Bob locally rotate their systems and measure them. 
Each observer locally branches into two versions of themselves, recording outcomes~`0' and~`1' with multiversal measures~$(1/2,1/2)$
Physical systems that testify to the violation of Bell’s inequality emerge only when Alice and Bob interact to compare results—that is, when a joint record is generated.
This joint record locally branches into four versions, corresponding to the four possible pairs $(00,01,10,11)$. The multiversal measures assigned to those records precisely match the quantum statistics: in the CHSH test, the winning pairs sum to~$\cos^2(\pi/8)$.

As in teleportation (see Sec.~\ref{sec:teleportation}), the communication that enables the comparison of results is “classical” only in the quantum-theoretic sense: robust under decoherence and implementable as a chain of local interactions. Bell experiments thus illustrate not nonlocal coordination at a distance, but instead a phenomenon that escapes single-world logic: joining records is not a trivial operation in the multiverse.
The branches with multiversal measures~$(1/2,1/2)$ combine in a nontrivial way when assembled into joint lists.
See Ref.~\cite{bedard2025explaining} for a full analysis.

\section{Discussion}\label{sec:disc}

The many-to-one correspondence between the universal descriptor and the global Schrödinger state was noticed by Timpson~\cite{timpson2005nonlocality}, and further studied with Wallace in Ref.~\cite{wallacetimpson07}. They argued that since descriptors corresponding to the same Schrödinger state lead to the same observations, they should be identified by a `quantum gauge equivalence'. In this case, the additional descriptor structure is discarded, and one is left with the usual Schrödinger state, thereby retrieving the familiar `nonlocality of states'. In contrast, Raymond-Robichaud~\cite{raymond2021local}, who also emphasized the \emph{non-injectivity} of the morphism between noumenal states (descriptors) and phenomenal states (density matrices), rejects the Wallace–Timpson identification. He treats noumenal states as elements of reality in their own right, thereby restoring locality even when distinct noumenal states give rise to the same observations.

Rejecting the Wallace–Timpson identification risks the charge of metaphysical indulgence, since it posits entities not fixed by the Schrödinger state alone. The first reply is that the boundary between metaphysical and physical shifts with the growth of knowledge. For instance, atomism, the corpuscular theory of light and the theory of terrestrial motion were branded speculative metaphysics before they became testable science. 
%
But there is a stronger reply available. 
Deutsch argues that the Wallace–Timpson gauge would eliminate physically meaningful distinctions concerning dynamics, causation, and explanation~\cite{deutsch2011vindication}. 
Descriptor differences are indeed invisible if one restricts attention to what is observable on a single system at a fixed given time. 
Yet once the underlying dynamical process is probed---by varying inputs, prepending gates, or inserting suitable measurement interactions---the dynamics can be identified, and with them the corresponding descriptor histories. 
Thus if descriptors are underdetermined by a fixed-time Schrödinger state, they are not underdetermined by physical processes.
The descriptor surplus is not arbitrary metaphysical decoration, it is dynamical structure. 

Relatedly, a reason I am putting forth for opposing the Wallace–Timpson identification is that descriptors allow us to solve important problems, such as those laid out in~\S\ref{sec:expl}.
We ought, therefore, to take them seriously.
If descriptors supply explanatory resources unavailable in the Schrödinger picture, quotienting them away would not be a harmless gauge choice. 
It would discard precisely what does the explanatory work.
The point would only become starker with further progress.
%
%

Finally, there is little reason to regard the axioms from which we commonly draw the instrumentalist equivalence as the final word.
Notably, \textbf{A4} posits a linear functional between states and observables, which fixes the observational predictions.
%
%
Advocates of Everettian quantum mechanics have rightly been dissatisfied with such a bare axiomatic rule, and have explored many routes 
to explain the expectation-value calculus and its associated probabilities.
For the Heisenberg programme, the key challenge is to understand the status of the Heisenberg state: what exactly is it, and why should it appear in expectation values at all?
Only if we take~\textbf{A4} as definitive and regard equality of expectation values as exhausting the physical content of a description can the usual equivalence of the Heisenberg and Schrödinger pictures be sustained. 
But surely that cannot be the whole story.


\subsection*{Acknowledgements}
I am grateful to Jacob Barandes, David Deutsch, Samuel Kuypers, Alyssa Ney, Simon Saunders, Christopher Timpson, Lev Vaidman, and Vlatko Vedral for stimulating discussions and feedback on earlier versions of this chapter.
I am also grateful to the Conjecture Institute for welcoming me as a fellow and for its intellectual support, particularly through the careful feedback of Logan Chipkin.

This work was supported 
by
the Mitacs Elevate postdoctoral fellowship in
partnership with Bbox Digital.


\end{document}